\documentclass[conference]{IEEEtran}

\usepackage{amsmath,amssymb,mathrsfs}
\usepackage{enumitem}
\usepackage{amsthm}
\usepackage{algcompatible}
\usepackage{algorithm}
\usepackage{algpseudocode}
\usepackage[pdftex]{graphicx}
\usepackage{tabularx}
\usepackage{tikz}
\definecolor{cof}{RGB}{219,144,71}
\definecolor{pur}{RGB}{186,146,162}
\definecolor{greeo}{RGB}{91,173,69}
\definecolor{greet}{RGB}{52,111,72}
\usetikzlibrary{patterns}
\usetikzlibrary{arrows}
\usepackage{pgfplots}
\usepackage{multirow}
\usepackage{hhline}
\usepackage{bbm}
\newtheorem{lem}{Lemma}
\newtheorem{thm}{Theorem}
\newtheorem{dfn}{Definition}

\begin{document}
%
\title{Quality-of-Service in Multihop Wireless Networks: Diffusion Approximation}

\author{\IEEEauthorblockN{Ashok Krishnan K.S. and Vinod Sharma}
\IEEEauthorblockA{Dept. of ECE, Indian Institute of Science, Bangalore, India\\
Email: \{ashok, vinod\}@iisc.ac.in
}}

\maketitle

\begin{abstract}
	We consider a multihop wireless system.  There are multiple source- destination pairs.  The data from a source may have to pass through multiple nodes. We obtain a channel scheduling policy which can guarantee end-to-end mean delay for the different traffic streams. We show the stability of the network for this policy by convergence to a fluid limit. It is intractable to obtain the stationary distribution of this network. Thus we also provide a diffusion approximation for this scheme under heavy traffic. We show that the stationary distribution of the scaled process of the network  converges to that of the Brownian limit. This theoretically justifies the performance of the system. We provide simulations to verify our claims.
\end{abstract}

%
\IEEEpeerreviewmaketitle

\section{Introduction and Literature Review}
A multihop wireless network is constituted by nodes communicating over a wireless channel. Some of the nodes, called source nodes, have data to be sent to other nodes, called receivers. In general, the data will have to be transmitted across multiple hops, over other nodes. It is necessary to develop algorithms that can ensure transmission of these data packets across the network. Any such algorithm has to take into account the topology of the network and the variability of the channels. Further, different types of data, originating from different applications, may have different Quality-of-Service (QoS) requirements, such as delay or bandwidth constraints. To design algorithms that can meet all these requirements is of interest. It is also of interest to demonstrate the performance of these algorithms in theory and by simulations.\\
\indent The characterization of network performance has been approached at different angles, using various mathematical techniques. Stability of flows in a network is a common QoS requirement. Algorithms based on backpressure, such as in in \cite{NeelyJsac}, are throughput optimal, which means that they stabilize  the network if it is possible by any other policy. Another approach is to use the framework of Markov Decision Processes \cite{PRKumar}.\\
\indent The analysis of fluid scaling of networks was pioneered in works such as \cite{RybkoStolyar} and \cite{Dai95}, where it was demonstrated that stability of the fluid limit of the network implies the stability of the network. Further, one may obtain bounds on moments of asymptotic values of the queues using these techniques \cite{DaiMeyn}.  A comprehensive treatment of work in this direction is provided in \cite{MeynBook}.\\
\indent Diffusion approximation of networks \cite{Williams2} study the behaviour of the system under a scaling corresponding to the Functional Central Limit Theorem \cite{BillBook}. The weak limit of the diffusion scaled systems under heavy traffic is generally a reflected Brownian motion \cite{HarrisonBook}, which under certain assumptions on the scaling rate, has a limiting stationary distribution. This distribution may be used as a proxy for the actual distribution of the system state. The diffusion approximation of the Maxweight algorithm is studied in \cite{StolyarSSC}, using properties of certain fluid scaled paths to obtain properties of the diffusion scaled paths, as in \cite{BramsonSSC}. Of these, \cite{StolyarSSC} deals with a discrete time switch under the MaxWeight policy.\\
\indent To further justify the use of the Brownian limit as a proxy for the actual system, one may try to obtain conditions in which the scaling and time limits may be interchanged. Sufficient conditions for the same are studied in \cite{GamarZeevi} and \cite{Budhiraja}, in the case of Jackson Networks. An important requirement for the exchange of limits in \cite{Budhiraja} to hold is the Lipschitz continuity of an underlying Skorohod map, which may not always hold in general.\\
Our main contributions in this work are summarized below.
\begin{itemize}
	\item We propose an algorithm that solves, in every slot, a weighted optimization problem. Using time varying weights that are functions of the queue lengths and mean delay requirements, the algorithm is able to dynamically cater to mean delay requirements of different flows.
	The function being optimized is the same as in \cite{PaperICC}. However, the optimization here is in every slot, and does not use the technique of discrete review. The performance of these two algorithms are same from the point of view of throughput optimality, since both result in the same set of fluid equations, and consequently are both throughput optimal.
	\item We obtain a reflected Brownian motion (with drift) as the weak limit of the system under diffusion scaling, using techniques similar to \cite{StolyarSSC}. This Brownian motion exhibits state space collapse. 
	\item We also show that the stationary distribution of our network converges to the stationary distribution of the limiting Brownian network. This allows us to obtain the stationary distribution of our network by that of the limiting network which is explicitly available. However, our proof does not require Lipschitz continuity of the Skorohod map, unlike \cite{Budhiraja}.
\end{itemize} 
The rest of the paper is organized as follows. In Section II, we describe the system model and formulate the control policy used in the network. In Section III, we describe the two scaling regimes in which we study the network. In Section IV we prove the existence of the Brownian limit, and in Section V we prove that the stationary distribution of the limit of the scaled process is the stationary distribution of the limiting Brownian process.
\section{System Model and Control Policy}
We consider a multihop wireless network (Fig. \ref{fig1}). The network is a connected graph $\mathcal G=(\mathcal V,\mathcal E)$ with $\mathcal V=\{1,2,..,N\}$ being the set of nodes and $\mathcal E$ being the set of links on $\mathcal  V$. The system evolves in discrete time denoted by $t\in \{0,1,2,...\}$. The links are directed, with  link $(i,j)$ from node $i$ to node $j$ having a time varying channel gain $H_{ij}(t)$ at time $t$. Denote the channel gain vector at time $t$ by $H(t)$, evolving as independent and identically distributed (i.i.d.) process across slots with distribution $\gamma$ over a finite set $\mathcal H$. Let $E_h(t)$ denote the cumulative number of slots till time $t$ when the channel state was $h\in\mathcal{H}$. Let the vector of all $E_h(t)$ be denoted by $E(t)$.

 At a node $i$, $A_i^f(t)$ denotes the cumulative (in time) process of exogenous arrival of  packets destined to node $f$. The packets arrive as an  i.i.d sequence across slots, with mean arrival rate  $\lambda_i^f$ and variance $\sigma_i^f$. Let $\lambda$ denote the vector of all $\lambda_i^f$. All traffic in the network with the same destination $f$ is called  \emph{flow} $f$; the set of all flows is denoted by $\mathcal F$. Each flow has a fixed route to follow to its destination. At each node there are queues, with $Q_i^f(t)$ denoting the queue length at node $i$ corresponding to flow $f \in \mathcal F$ at time $t$. For a queue $Q_i^f$ with $i\neq f$, we have the queue evolution given by,
\begin{align}
Q_i^f(t)= Q_i^f(0)+A_i^f(t)+R_i^f(t)-D_i^f(t), \label{actualQueue}
\end{align}
where $R_i^f(t)$ is the cumulative arrival of packets by routing (i.e., arrivals from other nodes), and $D_i^f(t)$ is the cumulative departure of packets. Let us denote, by $S_{ij}^f(t)$, the cumulative number of packets of flow $f$ transmitted over link $(i,j)$. We write,
\begin{align}
R_i^f(t)=\sum_{k\neq i}S_{ki}^f(t),\text{ and } D_i^f(t)=\sum_{j \neq i}S_{ij}^f(t). 
\end{align}
 We assume that the links are sorted into $M$ \emph{interference sets}\ $I_1, I_2, \dotsc, I_M$. At any time, only one link from an interference set can be active. A link may belong to multiple interference sets.  We also assume that each node transmits at unit power. Then, the rate of transmission between node $i$ and node $j$  is given by an achievable rate function of $H(t)$ and $I(t)\in\{I_1, \dotsc I_M\}$, the schedule at time $t$.
 
The vector of queues at time $t$ is denoted by $Q(t)$. Similarly we have the vectors $A(t)$, $R(t)$, $D(t)$ and $S(t)$.

 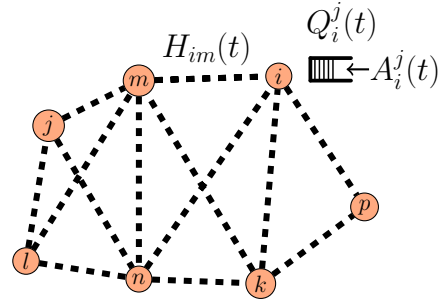
\begin{figure}
 	\centering
 	\setlength{\unitlength}{1cm}
 	\thicklines
 	\begin{tikzpicture}[scale=0.6, transform shape]		
 	\node[draw,shape=circle, fill={rgb:orange,1;yellow,0;pink,2;green,0}, scale=0.6, transform shape] (v1) at (4.7,0.5) {\Huge $k$};
 	\node[draw,shape=circle, fill={rgb:orange,1;yellow,0;pink,2;green,0}, scale=0.6, transform shape] (v2) at (2.0,0.6) {\Huge$n$};
 	\node[draw,shape=circle, fill={rgb:orange,1;yellow,0;pink,2;green,0}, scale=0.6, transform shape] (v3) at (2.0,5.0) {\Huge$m$};
 	\node[draw,shape=circle, fill={rgb:orange,1;yellow,0;pink,2;green,0}, scale=0.6, transform shape] (v4) at (0.0,4) {\Huge$j$};
 	\node[draw,shape=circle, fill={rgb:orange,1;yellow,0;pink,2;green,0}, scale=0.6, transform shape] (v5) at (5.1,5.1) {\Huge$i$};
 	\node[draw,shape=circle, fill={rgb:orange,1;yellow,0;pink,2;green,0}, scale=0.6, transform shape] (v6) at (7,2.2) {\Huge$p$};
 	\node[draw,shape=circle, fill={rgb:orange,1;yellow,0;pink,2;green,0}, scale=0.6, transform shape] (v10) at (-0.5,1.0) {\Huge$l$};
 	\node (v7) at (7.9,5.25) {\huge $A_i^j(t)$};
 	\node (v8) at (6.5,6.3) {\huge $Q_i^j(t)$};
 	\node (v9) at (3.5,5.7) {\huge $H_{im}(t)$};			
 	\draw[line width=0.8mm, dashed] (v2) -- (v1)
 	(v4) -- (v2)
 	(v2) -- (v5)
 	(v3) -- (v5)
 	(v3) -- (v2)
 	(v1) -- (v5)
 	(v1) -- (v3)
 	(v4) -- (v3)
 	(v3) -- (v5)
 	(v5) -- (v6)
 	(v10) -- (v2)
 	(v10) -- (v4)
 	(v10) -- (v3)
 	(v1) -- (v6);
 	\draw[line width=0.5mm, line cap=round](5.8,5)--(6.7,5);
 	\draw[line width=0.5mm, line cap=round](5.8,5.5)--(6.7,5.5);
 	\draw[line width=0.5mm, line cap=round](5.8,5)--(5.8,5.5);
 	\draw[line width=0.2mm](5.9,5)--(5.9,5.5);
 	\draw[line width=0.2mm](6.0,5)--(6.0,5.5);
 	\draw[line width=0.2mm](6.1,5)--(6.1,5.5);
 	\draw[line width=0.2mm](6.2,5)--(6.2,5.5);
 	\draw[line width=0.2mm](6.3,5)--(6.3,5.5);
 	\draw[thick,->] (7.1,5.25) -- (6.6,5.25);
 	\end{tikzpicture}
 	\caption{A simplified depiction of a Wireless Multihop Network}
 	\label{fig1}
 \end{figure}
\indent We want to develop scheduling policies such that the different flows obtain their  end-to-end mean delay deadline guarantees. Define $Q_{ij}^f=\max(Q_i^f-Q_j^f,0),  Q^f(t)=\sum_i Q_{i}^f(t)$, and let $\mathcal{M}(t)$ be the set of feasible rates at time $t$, which depends on $H(t)$. Our network control policy is as follows.  
At each $t$, we obtain the optimal allocation $\mu^*$,
\begin{align}
{\mu}^*=\arg_{\mu\in\mathcal{M}(t)}\max\sum_{i,j,f}\alpha(Q^f(t),\overline{Q}^f) Q_{ij}^f(t)\mu_{ij}^f, \label{optFun}
\end{align}
assuming $Q_{ij}^f>0$ for at least one link flow pair $(i,j),f$. If all $Q_{ij}^f$ are zero, we define the solution to be $\mu^*=0$. We optimize a weighted sum of rates, with more weight given to  flows with larger backlogs, with $\alpha$ capturing the delay requirement of the flow.  The weights $\alpha$  are functions of $Q^f(t)$, and $\overline{Q}^f$ denotes a desired value for the queue length of flow $f$. We use the function
\begin{align}
\alpha(x, \overline x)=1+\frac{a_1}{1+\exp (-a_2(x-\overline x))}. \label{logiF}
\end{align}
Thus, flows requiring a lower mean delay would have a higher weight compared to flows needing a higher mean delay. Flows whose mean delay requirements are not met should get priority over the other flows.The $\overline{Q}^f$ are chosen, using Little's Law, as  $\overline{Q}^f=\lambda^f\overline D$, where $\overline D$ is the target end to end mean delay and $\lambda^f$ is the arrival rate of flow $f$. 

Let $G_{ijf}^{hI}(t)$ be the number of slots till time $t$, in which channel state was $h$, the schedule was $I$ and flow $f$ was scheduled over $(i,j)$. Denote the vector of all $G_{ijf}^{hI}(t)$ by $ G(t)$.  Define the process,
\begin{align}
Z=({A},{E},{G},{D},{R},{S},{Q}), \label{ProcessDefn}
\end{align}
where we have $A =({A}(t),t\geq 0)$ (and likewise for the other processes). This process describes the evolution of the system.  The state of the system at time $t$ is ${Q}(t)$, which takes values in a state space $\mathcal{Q}$. Define the \emph{capacity region} as follows.
\begin{dfn}
	The capacity region $\Lambda$ of the network is the set of all $\lambda$ for which a stabilizing policy exists.
\end{dfn}
\subsection{Notational Convention}
We denote the set of real numbers by $\mathbb{R}$, and the set of integers  by $\mathbb{Z}$.
We  use $\mathscr C[0,\infty)$ to denote the set of all continuous functions from $[0,\infty)$ to $\mathbb{R}$, and $\mathscr D[0,\infty)$ the set of all right continuous functions with left limits (RCLL) from $[0,\infty)$ to $\mathbb R$. We  use $\implies$ to denote weak convergence. For a vector $x$, $|x|$ denotes its norm (modulus). The vector of variables of the form $x_i^j$ over all $i$ and $j$ will be denoted by $(x_i^j)_{i,j}$.\\
The list of symbols used in this paper  is summarized below, in Table \ref{tab:1}.
\begin{table}[h]
	\centering
	\begin{tabular}{|l|l|}
		\hline
		$\mathcal V$ & \textcolor{black}{Set of nodes}\\
		\hline
		$\mathcal E$ & \textcolor{black}{Set of Edges}\\
		\hline
		$\mathcal H$ & \textcolor{black}{Set of Channel States}\\
		\hline
		$\mathcal F$ & \textcolor{black}{Set of Flows}\\
		\hline
		$Q_i^f$ & \textcolor{black}{Queue Length of flow $f$ at node $i$}\\
		\hline 
		$A_i^f$  & \textcolor{black}{Cumulative Exogenous Arrivals to $Q_i^f$}\\
		\hline
		$D_i^f$  & \textcolor{black}{Cumulative Departures from $Q_i^f$}\\
		\hline
		$R_i^f$ & \textcolor{black}{Cumulative Arrivals to $Q_i^f$ by routing}\\
		\hline
		$S_{ij}^f$ & \textcolor{black}{Cumulative number of packets of flow $f$ served on link $(i,j)$}\\
		\hline
		$H_{ij}$ & \textcolor{black}{Channel gain across link $(i,j)$}\\
		\hline
		$E_h$ & \textcolor{black}{Cumulative slots when channel gain was $h$}\\
		\hline
		$G_{ijf}^{hI}$ & \textcolor{black}{Time with channel $h$, schedule $I$, flow $f$  scheduled on $(i,j)$}\\
		\hline	 
		$Z(t)$ & The process $(A(t), E(t), G(t),D(t), R(t), S(t), Q(t), t\geq 0)$\\
		\hline	 
		$Z^n(t)$ & The process corresponding to $n$-th scaled system\\
		\hline
		$\lambda^n$ & Arrival rate of $n$-th system\\
		\hline
		$\psi$ & Normal vector at boundary of capacity region\\
		\hline
		$W^n(t)$ & Workload $=\langle\psi, Q^n(t)\rangle $\\
		\hline	
		$z^n(t)$ & The process $Z^n(\lfloor nt\rfloor )/n$\\
		\hline	
		$\hat z^n(t)$ & The process $Z^n(\lfloor n^2t\rfloor)/n$\\
		\hline
	\end{tabular}
	\caption{\textcolor{black}{List of Symbols}}
	\label{tab:1}
\end{table}

\section{Two Scaling regimes}
 Now we describe the behaviour of $Z$ under two scaling regimes, \emph{fluid} and \emph{diffusion}.
\subsection{Fluid Scaling}
For the process $Z$, define the scaled continuous time process,
\begin{align}
z^n(t) &=\frac{Z(\lfloor nt\rfloor )}{n}, \label{FluScaEq}
\end{align}
where $\lfloor\cdot \rfloor$ represents the floor function. This is called the fluid scaled process. Note that the time argument $t$ on the left side is continuous, while that on the right is discrete. Whether a time argument is discrete or continuous will be generally clear from the context. Let $z^n$ denote the process $(z^n(t),t\geq 0)$. We have,
\begin{align}
z^n=(a^n,e^n,g^n,d^n,r^n,s^n,q^n),
\end{align}
with the scaling in (\ref{FluScaEq}) being applied to each component of $Z$. Note that $a^n=(a_i^{f,n})_{i,f}$, and a similar notational convention holds for all the constituent functions of $z$. The limit of $z^n$, as $n\to\infty$, offers insight into the behaviour of the system under the scheduling policy in (\ref{optFun}). The following result may be shown for our policy.
\begin{lem}
The algorithm described by the slotwise optimization in (\ref{optFun}) stabilizes the system for all arrival rate vectors $\lambda$ in the interior of $\Lambda$. Here, stability implies that the Markov chain $Q(t)$ is positive recurrent.	
\end{lem}
To prove this, we first show that, almost surely, a subsequential limit exists for the family $\{z^n,n\geq 0\}$. This limit $z$ is called the fluid limit, which obeys a deterministic ordinary differential equation (o.d.e.). The proof follows by showing that this o.d.e. is globally asymptotically stable, by constructing a suitable Lyapunov function. The stability of the o.d.e. implies the stability of the associated stochastic process. The detailed proof is similar to that in \cite{PaperICC}, and the algorithm here and in \cite{PaperICC} will have the same fluid limit equations.\\
\indent Studying the fluid limit gives us insights into the stability properties of the system. However, it only proves the existence of a stationary distribution. In order to predict the behaviour of the system, one needs the stationary distribution, or some approximation to the same.  However, explicitly computing the stationary distribution for our system is not feasible. Thus we define the heavy traffic regime, and the associated diffusion scaling, below. We will also show that the stationary distribution of our system process converges to that of the limiting Brownian network. This will provide us an approximation of the stationary distribution under heavy traffic, the scenario of most practical interest.
\subsection{Diffusion Scaling}
Consider a sequence of systems, $Z^n$. Each system differs from the other in its arrival rate, $\lambda^n$. The $\lambda^n$ are chosen such that, as $n\to\infty$, $\lambda^n\to\lambda^*$, and,
\begin{align}
\lim_{n\to\infty}n\langle\psi,\lambda^n-\lambda^*\rangle=b^*\in\mathbb{R}, \label{cond1}
\end{align}
where $\lambda^*$ is a point on the boundary of $\Lambda$, and $\psi$ denotes the outer normal vector to $\Lambda$ at the point $\lambda^*$. This is known as heavy traffic scaling. We will also assume that $\lambda^*$  falls in the relative interior of one of the faces of the boundary of  $\Lambda$. For this sequence of systems, we define the diffusion scaling, given by,
\begin{align}
\hat z^n(t)=\frac{Z^n(\lfloor n^2t\rfloor)}{n}. \label{DiffScalDefn}
\end{align}
Let $\hat z^n$ denote the process $(\hat z^n(t),t\geq 0)$. As before, we have,
\begin{align*}
\hat z^n=(\hat a^n,\hat e^n,\hat g^n,\hat d^n,\hat r^n,\hat s^n,\hat q^n).
\end{align*}
 
Define the system workload $W^n(t)$ in the direction $\psi$ as,
\begin{align}
W^n(t)=\langle\psi,Q^n(t)\rangle.
\end{align}
Define the scaled process $\hat w^n=(\hat w^n(t),t\geq 0)$ by,
\begin{align*}
\hat w^n(t)=\frac{W(\lfloor n^2 t\rfloor)}{n}.
\end{align*}
 Define an invariant point to be a vector $\phi$ that satisfies, for some $k>0$,
\begin{align}
\alpha(\phi)\phi=k\psi,
\end{align}
where $\alpha(\phi)$ is the vector of all $\alpha(\phi_j)$. Then, we have the following result, which characterizes the weak convergence of the diffusion scaled processes.
\begin{thm}\label{CompleteTheorem}
	Consider a sequence  $\{\hat z^n,n\in\mathcal{N}\}$ as described above, under heavy traffic scaling satisfying (\ref{cond1}),and $\mathcal{N}$ a sequence of positive integers $n$ increasing to infinity. Assume that the fluid scaled $z=(a,e,g,d,r,s,q)$ has components $a=(a_i^f)_{i,f}$ and $e=(e_h)_{h\in\mathcal H}$ that satisfy, with probability one, as $m\to\infty$, for any $T>0$, for all $i$, $j$, $f$, $c\in\mathcal{H}$, 
		\begin{align}
		&\max_{0\leq \ell\leq nT}\sup_{0\leq \epsilon\leq 1}|a_i^{f,n}(\ell+\epsilon)-a_i^{f,n}(\ell)-\lambda_i^f\epsilon|\to 0, \label{aWeakLim}\\
		&\max_{0\leq \ell\leq nT}\sup_{0\leq \epsilon\leq 1}|e_c^n(\ell+\epsilon)-e_c^n(\ell)-\gamma_c\epsilon|\to 0.\label{eWeakLim}
		\end{align}	
	Further, assume that,
	\begin{align}
	\hat q^n(0)\implies c\phi,
	\end{align}
	where $c$ is a non negative real number. Then, the sequence $\{\hat w^n,n\in\mathcal{N}\}$ converges weakly to a reflected Brownian motion $\hat w$ as $n\to\infty$. Further, $\{\hat q^n,n\in\mathcal{N}\}$ converges weakly to $\phi\hat w$.
	\end{thm}
	The proof of this Theorem is detailed in the following section.
	\section{Brownian Limit}
 The existence of the Brownian limit is demonstrated as follows. We write the scaled workload $\hat w^n$ as the sum of two terms, one of which converges to a \emph{free} Brownian motion, and the second as its corresponding \emph{regulating} process. Together, they act as a reflected Brownian motion. Let us define, for a channel state $h\in\mathcal H$, $\mathcal D_h$ as the set of all feasible rate vectors. Let us denote the maximum allocation in the direction $\psi$, when the channel is in state $h$, by $\mu_h$, 
\begin{align}
\mu_h =\max_{y\in\mathcal{D}_h}\langle {\psi},y \rangle, \ h\in\mathcal{H}.
\end{align}
Let $\mu_1$ be the vector $(\mu_1,\dotsc,\mu_{|\mathcal{H}|})$, and $\mu_2$ the vector $(\mu_1^2,\dotsc,\mu^2_{|\mathcal{H}|})$.
Define the random variables,
\begin{align*}
X^{\mu}(t)=\mu_{H(t)}, \ t\geq 1.
\end{align*}
The random variables $\{X_{\mu}(t),t\geq 0\}$ are i.i.d, with mean  and variance given by,
\begin{align*}
\hat{\mu} &:=\langle\mu_1,\gamma\rangle,\\ 
\hat{\sigma}^2 &:=\mathbb{E}[(X_{\mu}(1)-\hat{\mu})^2] =\langle\mu_2,\gamma\rangle-\hat{\mu}^2\geq 0. \label{varianceEqn}
\end{align*}
Define the cumulative process,
\begin{align}
X(t) &= \sum_{k=1}^t X^{\mu}(k).
\end{align}
This is the cumulative maximum possible service in the direction $\psi$. We can write,
\begin{align}
{U}(t)&=W(0)+\langle {\psi},{A}(t)\rangle-X(t),\\
V(t)&=X(t)+\langle {\psi},{R}(t)\rangle
-\langle {\psi},{D}(t)\rangle,
\end{align}
and, consequently,
\begin{align}
W(t)=U(t)+V(t).
\end{align}
The same equation holds for  $W^n$, $U^n$ and $V^n$. Define,
\begin{align*}
\hat u^n(t)=\frac{U^n(\lfloor n^2t\rfloor)}{n},\  \hat v^n(t) &=\frac{V^n(\lfloor n^2t\rfloor)}{n}.
\end{align*}
Thus we have,
\begin{align}
\hat w^n(t)=\hat u^n(t)+\hat v^n(t).
\end{align}
Let us denote $\hat w^n=(\hat w^n(t),t\geq 0)$, $\hat u^n=(\hat u^n(t),t\geq 0)$ and $\hat v^n=(\hat v^n(t),t\geq 0)$.
We have the following result about the convergence of $\{\hat u^n,n\geq 0\}$.
\begin{lem}\label{uTheorem} Assuming that the initial condition converges weakly to an invariant point, i.e,
	\begin{align}
	\hat{w}^n(0)\implies\hat{w}(0),
	\end{align}
	where $\alpha(\hat{w})\hat{w}=\psi$. Then, it follows that,
	\begin{align}
	\hat{u}^n\implies\hat{u},
	\end{align}
	in $\mathscr D[0,\infty)$ where $\hat u=(\hat u(t),t\geq 0)$ is a Brownian motion with drift, given by,
	\begin{align}
	\hat{u}(t)=\hat{w}(0)+b^*t+\sigma \mathscr{B}(t), \label{BrownEqn}
	\end{align}
	where $\mathscr{B}(t)$ is a standard Brownian motion, $\sigma^2=\sum_{i,f}(\sigma_i^f)^2+\hat{\sigma}^2$, and $b^*$ is given by (\ref{cond1}).
\end{lem}
\begin{proof}
	This is an application of Donsker's theorem \cite{BillBook}. We can write $\hat u^n$ as,
	\begin{align*}
	\hat u^n(t) &=\frac{U^n(n^2t)}{n}=\hat  w^n(0)+\langle\psi, \hat a^n(t)\rangle-\hat x^n(t),\\
	&=\hat  w^n(0)+\langle\psi, \hat a^n(t)-\lambda^n nt\rangle-(\hat x^n(t)-\hat{\mu} nt)\\
	&\ \ \ \ \ \ \ \ \ \ \ \ \ +(\langle\psi,\lambda^n\rangle-\hat{\mu}) nt.
	\end{align*}
	Since $\hat{\mu}=\langle\psi,\lambda^*\rangle$, from assumption (\ref{cond1}), it follows that,
	\begin{align*}
	(\langle\psi,\lambda^n\rangle-\hat{\mu}) nt\to b^*t.
	\end{align*}
	The convergence of the processes $(\langle \psi,\hat a^n(t)-\lambda^n nt\rangle,t\geq 0)$ and $(\hat x^n(t)-\hat{\mu} nt,t\geq 0)$ to independent Brownian motions, by Donsker's theorem, now implies the result.
\end{proof}
 
Now we outline the proof of Theorem \ref{CompleteTheorem}.

\begin{proof}[Proof of Theorem \ref{CompleteTheorem}]
	From Lemma \ref{uTheorem}, using the Skorohod representation Theorem \cite{BassBook}, one can construct a probability space where we have $\mathscr C[0,\infty)$ valued processes $\hat u_{S}^n$ and $\hat u_S$,  such that, almost surely,
	\begin{align*}
	\hat u_S^n\to\hat u_S\ u.o.c.,
	\end{align*}
	where $\hat u_S^n$ and $\hat u_S$ are identical in distribution to $\hat u^n$ and  $\hat u$. Thus $\hat u_S$ is the Brownian motion given in (\ref{BrownEqn}). We augment this probability space to include the other components of $Z$ as well. On this probability space, we will have the functions $\hat v^n$ and $\hat w^n$ as before. Note that, almost surely, for any sequence of $n$ increasing to infinity, properties (\ref{aWeakLim}) and (\ref{eWeakLim}) hold \cite{StolyarSSC}.
	
	The convergence of $\{\hat w^n,n\in\mathcal{N}\}$ now weakly follows if we show that, for any subsequence $\mathcal N_1$ of $\mathcal N$, there exists a subsequence $\mathcal N_2$, such that, as $n\to\infty$ along $\mathcal N_2$, almost surely,
	\begin{align}
	\hat v^n\to\hat v,\ u.o.c.,
	\end{align}
	where, almost surely, $\hat v(t)$ is continuous and finite  for $t\in[0,\infty)$, $\hat{v}(0)=0$ and if $\hat{w}(t)>0$, then $t$ is not a point of increase of $\hat{v}(t)$. Then, it can be shown \cite{HarrisonBook} that  $\hat v$ is unique, and called the \emph{regulator} corresponding to $u$, and can be represented as,
\begin{align}
\hat v(t)=-\inf_{0\leq s\leq t}u(s).
\end{align}
Consequently, $\hat w(t)=\hat u(t)+\hat v(t)\geq 0$.  This $w(t)$ is called the \emph{reflected} or \emph{regulated} Brownian motion corresponding to $u$. Then it follows that $\hat w^n$ converges weakly to a reflected Brownian motion as $n\to\infty$.\\
Therefore it suffices to show that $\hat v^n$ has a limit $\hat v$ along $\mathcal{N}_2$  which satisfies the requisite properties. This is proven in the Appendix. These properties also imply that, $\hat q^n$ converges weakly to $\phi\hat w$.
\end{proof}
Now that we have established the existence of a limiting Brownian motion, we proceed to demonstrate how the stationary distribution of the limit of the scaled systems is equivalent to tha stationary distribution of the Brownian motion, in the next section.
\section{Exchange of Limits}
We have the following result.
	\begin{thm}\label{LimitExchangeTheorem}
		The stationary distribution of the limiting process is the limit of the stationary distributions of the constituent processes, i.e.,
		\begin{align}
		\hat q^n(\infty)\implies\phi\hat w(\infty), \text{ as } n\to\infty,
		\end{align}
		where the time argument being infinity denotes the respective stationary distributions.
		\end{thm}
To prove this result, we first define a new set of fluid limit processes, given by,
\begin{align}
\bar z^{n,r}(t)=\frac{Z^n(\lfloor rt\rfloor)}{r}.
\end{align}
Let $\bar z^{n,r}=(\bar a^{n,r},\bar e^{n,r},\bar g^{n,r},\bar d^{n,r},\bar r^{n,r},\bar s^{n,r},\bar q^{n,r})$, denote the process $(\bar z^{n,r}(t),t\geq 0)$, and $\bar z^n$ the fluid limit process obtained, for each $n$, by taking the limit $r\to\infty$. This limit exists just as in the previous section. For each $Z^n$, let $\pi_n$ denote the stationary distribution of the queues. These exist because for each $n$, the system $Q^n$ is stable \cite{PaperICC}. The draining time (time for all queues to reach level zero) for the $n$-th fluid system will be denoted by $T_n$. From \cite{PaperICC}, we can see that $T_n$ is inversely proportional to the distance from the boundary of the capacity region $\Lambda$. It is also easy to see that, due to (\ref{cond1}), the distance to the boundary of the capacity region, which is the plane whose normal vector is $\psi$, grows as $\frac{1}{n}$. Hence we may write,
\begin{align}
T_n\leq nT_1, \label{DrainTimeEquation}
\end{align}
for some finite $T_1$, assuming that the initial fluid level is unity.\\
Now, we state a sufficient condition for the sequence $\{\pi_n, n\geq 0\}$ to be tight. Note that by writing $\hat q_x^n(\cdot)$ we indicate that the initial condition of the queue is $x$.
\begin{lem}
	Assume that, for all nodes $i$, $j$, flows $f$, for any $n\geq 1$, $t\geq 0$, we have, for some $B<\infty$,
	\begin{align}
	\mathbb{E}[\sup_{0\leq k \leq t}|A_i^{f,n}(k)-\bar a_i^{f,n}(k)|^2]\leq Bt,\label{Abound}\\
	\mathbb{E}[\sup_{0\leq k \leq t}|R_i^{f,n}(k)-\bar r_i^{f,n}(k)|^2]\leq Bt,\label{Rbound}\\
	\mathbb{E}[\sup_{0\leq k \leq t}|D_i^{f,n}(k)-\bar d_i^{f,n}(k)|^2]\leq Bt.\label{Dbound}
	\end{align}
	Further, assume that there exists $T$ such that for all $t\geq T$, we have,
	\begin{align}
	\lim_{|x|\to\infty}\sup_n\frac{1}{|x|^2}\mathbb{E}|\hat{{q}}_x^n(t|x|)|^2=0. \label{compactLim}
	\end{align}
	Then the sequence of distributions $\{\pi_n\}$ is tight.
\end{lem}
This result is a consequence of Theorems 3.2, 3.3 and 3.4 of \cite{Budhiraja}. We show that the conditions of this theorem hold in our case.

\begin{lem}
	In our system model, conditions (\ref{Abound})-(\ref{Dbound}) hold. Further, there exists $T$ such that (\ref{compactLim}) holds. Consequently, the sequence $\{\pi_n\}$ is tight.
\end{lem}
\begin{proof}
	Since the process $\{A_i^{f,n}(t)-a_i^{f,n}(t),t\geq 0\}$ is a martingale, we can use Doob's inequality \cite{BassBook} to obtain,
	\begin{align*}
	\mathbb{E}[\sup_{0\leq k \leq t}|A_i^{f,n}(s)-\bar a_i^{f,n}(s)|^2] &\leq B_1^{'} \mathbb{E}|A_i^{f,n}(t)-\bar a_i^{f,n}(t)|^2,\\
	&\leq B_1^{'}t\mathbb{E}|A_i^{f,n}(1)-\bar a_i^{f,n}(1)|^2,\\
	&=B_1 t,
	\end{align*}
	where the second inequality follows from the i.i.d nature of the arrival process \cite{GutBook}. Hence,  (\ref{Abound}) holds.\\	
	The bounds for $R$ and $D$ would hold if a corresponding bound holds for the $S_{ij}^f$ processes. Let us call the slotwise allocation process as $\bar S_{ij}^f$, where,
	\begin{align*}
	S_{ij}^f(t)=\sum_{t^{'}=1}^{t}\bar S_{ij}^f(Q(t^{'}),H(t^{'})),
	\end{align*}
	since  $\bar S_{ij}^f$ depends on both the queue state at time $t$, and the channel state at time $t$. Let $\mathcal S$ be the set of possible values $S(t)$ can take. Since $\mathcal H$ is finite (and consequently, $\mathcal S$), there are only a finite set of mappings from $\mathcal H$ to $\mathcal S$. This set of mappings will be denoted by $\{\mathbb F_1,\dotsc,\mathbb F_{K_1}\}$. Each $S(Q(t),H(t))$ will take the value of one of these functions. It is easy to see that the state space of queues can be partitioned as,
	\begin{align}
	\mathcal Q=\cup_{m=1,\dotsc,K_1}\mathcal Q_m,
	\end{align}
	where, if $Q(t)\in\mathcal Q_m$, we have $S(Q(t),H(t))=\mathbb F_m(H(t))$, and the $\mathcal Q_m$ are disjoint. Now we can write,
	\begin{align}
	S_{ij}^f(t)=\sum_{t^{'}=1}^t\sum_{m=1}^{K_1}\mathbb F_m(H(t))\mathbf 1_{\{Q(t)=m\}},
	\end{align}
	where $\mathbf 1$ is the indicator function. We can further rewrite this as,
	\begin{align}
	S_{ij}^f(t)=\sum_{m=1}^{K_1}\sum_{k\in \hat T_m(t)}\mathbb F_m(H(k)),
	\end{align}
	where $\hat T_m(t)$ is the set of time slots till $t$ when the queue state was in $\mathcal Q_m$. Since the system is stationary, we can also obtain,
	\begin{align}
	s_{ij}^f(t)=\mathbb{E}[S_{ij}^f(t)].
	\end{align}
	Thus, we may write, with $\bar{\mathbb F}_m=\mathbb{E}[\mathbb F_m(H(1))]$,
	\begin{align*}
	|S_{ij}^f(t)-s_{ij}^f(t)|^2\leq B_2^{'} \sum_{m=1}^{K_1}\left|\sum_{k\in \hat T_m(t)}\mathbb F_m(H(k))-\bar{\mathbb F}_m\right|^2,
	\end{align*}
	where $B_2^{'}$ depends only on $K_1$. For any $m$, along $k\in\hat T_m(t)$, $\mathbb F_m(H(k))$ is an i.i.d sequence. Therefore, proceeding similar to what was done for $A$, we now obtain,
	\begin{align*}
	\mathbb{E}[\sup_{0\leq k\leq t}|S_{ij}^f(k)-s_{ij}^f(k)|^2]\leq B_2\mathbb{E}[\sum_m |\hat T_m(t)|]= B_2t,
	\end{align*}
	where the equality follows, since $\sum_m |\hat T_m(t)|=t$.
	 Hence the bounds hold for $R$ and $D$ as well. Hence (\ref{Abound})-(\ref{Dbound}) hold, choosing $B=\max\{B_1,B_2\}$.
	
	To show (\ref{compactLim}), observe that, for a particular queue $Q_i^f$, it follows from the queueing equation that,
	\begin{align*}
	{n}\hat{q}_i^{f,n}(t) &= Q_i^{f,n}(n^2t),\\
	&=Q_i^{f,n}(0)+A_i^{f,n}(n^2t)+R_i^{f,n}(n^2t)-D_i^{f,n}(n^2t).
	\end{align*}
	Subtracting on either side with the corresponding fluid queue $q_i^{f,n}(t^{'})$ at time $t^{'}=n^2t$, we obtain,
	\begin{align*}
	Q_i^{f,n}(n^2t)-\bar q_i^{f,n}(n^2t) &= Q_i^{f,n}(0)-\bar q_i^{f,n}(0)+A_i^{f,n}(n^2t)\\
	&\ \ \ -\bar a_i^{f,n}(n^2t)	+ R_i^{f,n}(n^2t)-r_i^{f,n}(n^2t)\\
	&\ \ \ -D_i^{f,n}(n^2t)+\bar d_i^{f,n}(n^2t).
	\end{align*}
	Hence, we have,
	\begin{align*}
	|Q_i^{f,n}(n^2t)-\bar q_i^{f,n}(n^2t)|^2 &\leq C(|Q_i^{f,n}(0)-\bar q_i^{f,n}(0)|^2\\
	&\ \ \ +|A_i^{f,n}(n^2t)-\bar a_i^{f,n}(n^2t)|^2\\
	&\ \ \ + |R_i^{f,n}(n^2t)-\bar r_i^{f,n}(n^2t)|^2\\
	&\ \ \ +|D_i^{f,n}(n^2t)-\bar d_i^{f,n}(n^2t)|^2).
	\end{align*}
	Choosing $Q_i^{f,n}(0)=\bar q_i^{f,n}(0)$, we obtain, using (\ref{Abound})-(\ref{Dbound}),
	\begin{align}
	\mathbb E|Q_i^{f,n}(n^2t)-\bar q_i^{f,n}(n^2t)|^2\leq C_2n^2t,
	\end{align}
	and hence it follows for the vector process $ Q$ as well, with a higher constant $C_2^{'}$,
	\begin{align}
	\mathbb E| Q^{n}(n^2t)- \bar q^{n}(n^2t)|^2\leq C^{'}_2n^2t.
	\end{align}
	From (\ref{DrainTimeEquation}), since the draining time of the fluid system $ \bar q^{n}$ with initial condition equal to one, $T_n\leq {n}T_1$, the fluid system with initial condition $x$, will be zero at any time greater than $T_n|x|$. Setting $t\geq T_1|x|$, and dividing by $n^2$, we get,
	\begin{align}
	\mathbb{E}|\hat{ q}_x^n(t|x|)|^2\leq C_2^{'}t|x|.
	\end{align}
	Since the bound is uniform over $n$, dividing by $|x|^2$ and taking $|x|\to\infty$ gives the result.	
\end{proof}
With this result, we are ready to prove Theorem \ref{LimitExchangeTheorem}.
\begin{proof}[Proof of Theorem \ref{LimitExchangeTheorem}]
	Since the $\pi_n$ are tight, any subsequence of $\pi_n$ has a convergent subsequence. Let such a limit point be $\pi^*$. Assume that the initial conditions $\hat Z^n(0)$ are distributed as $\pi_n$. Since the systems $\hat Z^n$ converge to a reflected Brownian motion (RBM), the initial condition of the RBM  $\hat w$ will have distribution $\pi^*$. Also, we have shown that finite dimensional distributions of $\hat z^n$ also converge to that of $\hat w$. In particular, $\hat z^n(t)$ weakly converges to $\hat w(t)$ for any $t\geq 0$. But the distribution of $\hat z^n(t)$ is $\pi_n$.  Thus distribution of $\hat w(t)$ is $\pi^*$ for each $t$. Hence $\pi^*$ is the stationary distribution of $\hat w$.
\end{proof}

The Brownian motion $\hat w$ obtained as the limit of $\hat w^n$ is a unidimensional Brownian motion reflected at zero, having drift $b^*$. If $\hat w(\infty)$ has the stationary distribution of $\hat w$, we have, if $b^*<0$,
\begin{align}
\mathbb{P}[\hat w(\infty)<y]=1-\exp(2b^*y/\sigma^2),
\end{align}
from \cite{HarrisonBook}.
\section{Numerical Simulations}
We verify the validity of our approximations on a star network topology (Figure \ref{fig2}). There are two arrival processes, one arriving at node $1$, with node $4$ as its destination. The other arrives at node $2$, with node $5$ as destination. We will also assume that two links which share a common node interfere with each other.
\begin{figure}[h]
	\centering
	\setlength{\unitlength}{1cm}
	\thicklines
	\begin{tikzpicture}[scale=0.6, transform shape]		
	\node[draw,shape=circle, fill={rgb:orange,1;yellow,0;pink,2;green,0}, scale=0.6, transform shape] (v1) at (0,4) {\Huge $1$};
	\node[draw,shape=circle, fill={rgb:orange,1;yellow,0;pink,2;green,0}, scale=0.6, transform shape] (v2) at (0,0) {\Huge$2$};
	\node[draw,shape=circle, fill={rgb:orange,1;yellow,0;pink,2;green,0}, scale=0.6, transform shape] (v3) at (2,2) {\Huge$3$};
	\node[draw,shape=circle, fill={rgb:orange,1;yellow,0;pink,2;green,0}, scale=0.6, transform shape] (v5) at (4,0) {\Huge$5$};
	\node[draw,shape=circle, fill={rgb:orange,1;yellow,0;pink,2;green,0}, scale=0.6, transform shape] (v4) at (4,4) {\Huge$4$};			
	\draw[line width=0.8mm, dashed] (v3) -- (v1)
	(v3) -- (v5)
	(v3) -- (v4)
	(v2) -- (v3);
	\end{tikzpicture}
	\caption{An Example Network}
	\label{fig2}
\end{figure}
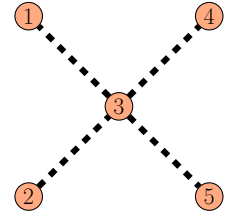
From the diffusion approximation and (\ref{StatDistbtnExp}), we can see that the mean of the Brownian motion corresponding to the queue can be approximated by the vector $\phi \frac{\sigma^2}{2b^*}$. The Brownian motion is a limit of the scaled process of the form $\frac{Q(n^2t)}{n}$. For a large $n$, we may approximately write,
\begin{align*}
Q(n^2t)\approxeq n\phi\frac{\sigma^2}{2b^*}.
\end{align*}
If we run the simulations for a time $n$, we may further also approximately write $b^*=n|\lambda-\lambda^*|$. Hence, we have the approximation,
\begin{align}
Q(\infty)\approxeq \phi \frac{\sigma^2}{2|\lambda-\lambda^*|}.
\end{align}
We assume that the channels are independent and identically distributed, with the distribution being uniform over the set $\{0,1,2,3\}$. We consider the arrival vector $(\lambda_1,\lambda_2)=(\lambda,\lambda)$, i.e., increasing along the line of unit slope. In this case $\lambda^*=(0.65,0.65)$. We will be looking at the total queue length of the flow $1\to3\to 4$. The value of $\sigma^2$ is $2\lambda+\hat{\sigma}^2$. The vector $\phi$ is approximately $(\frac{1}{\sqrt{2}},\frac{1}{\sqrt{2}})$ (The value of $\bar Q$ for both queues is set at $100$). We take $\hat{\sigma}^2\approxeq 8$. The values of the total queue length of the flow $1\to 3\to 5$ are listed in Table \ref{tab:2} (owing to symmetry both queue lengths are same), for simulation runs of length $10^5$, averaged over 20 simulations. It can be seen that the approximations follow the queue length closely. Moving within a small distance of the point $\lambda^*$ will require more iterations for the effects to show.\\
\begin{table}[h]
	\centering
	\begin{tabular}{|l|l|l|}
		\hline
		Arrival Rate $\lambda$ & Mean Queue Length & Approximation\\
		\hline
		0.64 & 233 & 232\\
		\hline
		0.641 & 263 & 258\\
		\hline
		0.642 & 319 & 290\\
		\hline
		0.643 & 367 & 332\\
		\hline
		0.644 & 381 & 387\\
		\hline
		0.645 & 479 & 465\\
		\hline
		0.646 & 517 & 581\\
		\hline
		0.647 & 568 & 775\\
		\hline
	\end{tabular}
	\caption{\textcolor{black}{Approximation of Queues}}
	\label{tab:2}
\end{table}
\indent In order to demonstrate that the algorithm can satisfy different  QoS requirements, we simulate the network at three points in the interior of the capacity region. The mean queue length asked from the flows is $250$ and $100$ respectively. We also pick $a_2$ in the expression of $\alpha$ for the second flow to be $4$, since it requires a tighter constraint to be met. In Table \ref{tab:3}, the first column gives the arrival rate, the second  shows the target queue length for the two flows, and the final column shows the queue length obtained. We see that the end-to-end mean queue length requirement is met for both the flows till rate $0.64$. At $0.641$ there is substantial departure. The capacity boundary is at $0.65$. Thus, our algorithm can provide QoS under heavy traffic as well.
\begin{table}[h]
	\centering
	\begin{tabular}{|l|l|l|}
		\hline
		$\lambda$ & Mean Queue Length Asked &  Queue Length Obtained\\
		\hline
		0.63 & (250,100) & (213,98)\\
		\hline
		0.64 & (250,100) & (264,110)\\
		\hline
		0.641 & (250,100) & (292,120)\\
		\hline
	\end{tabular}
	\caption{\textcolor{black}{Mean Queue Length Target and Obtained}}
	\label{tab:3}
\end{table}
\section{Conclusion}
We have presented an algorithm for scheduling in multihop wireless networks that guarantees end-to-end mean delays of the packets transmitted in the network. The algorithm is throughput optimal. Using diffusion scaling, we obtain the Brownian approximation of the algorithm. We also prove theoretically that the  stationary distribution of the limiting Brownian motion is the distribution of a sequence of scaled systems, and is consequently a good approximation for the stationary distribution of the original system. Using these relations, we obtain an approximation for queue lengths, and demonstrate via simulations that these are accurate.
\bibliography{survey}

\begin{thebibliography}{10}
\providecommand{\url}[1]{#1}
\csname url@samestyle\endcsname
\providecommand{\newblock}{\relax}
\providecommand{\bibinfo}[2]{#2}
\providecommand{\BIBentrySTDinterwordspacing}{\spaceskip=0pt\relax}
\providecommand{\BIBentryALTinterwordstretchfactor}{4}
\providecommand{\BIBentryALTinterwordspacing}{\spaceskip=\fontdimen2\font plus
\BIBentryALTinterwordstretchfactor\fontdimen3\font minus
  \fontdimen4\font\relax}
\providecommand{\BIBforeignlanguage}[2]{{%
\expandafter\ifx\csname l@#1\endcsname\relax
\typeout{** WARNING: IEEEtran.bst: No hyphenation pattern has been}%
\typeout{** loaded for the language `#1'. Using the pattern for}%
\typeout{** the default language instead.}%
\else
\language=\csname l@#1\endcsname
\fi
#2}}
\providecommand{\BIBdecl}{\relax}
\BIBdecl

\bibitem{NeelyJsac}
M.~J. Neely, E.~Modiano, and C.~E. Rohrs, ``Dynamic power allocation and
  routing for time-varying wireless networks,'' \emph{IEEE Journal on Selected
  Areas in Communications}, vol.~23, no.~1, pp. 89--103, 2005.

\bibitem{PRKumar}
R.~Singh and P.~Kumar, ``Throughput optimal decentralized scheduling of
  multi-hop networks with end-to-end deadline constraints: Unreliable links,''
  \emph{arXiv preprint arXiv:1606.01608}, 2016.

\bibitem{RybkoStolyar}
A.~N. Rybko and A.~L. Stolyar, ``Ergodicity of stochastic processes describing
  the operation of open queueing networks,'' \emph{Problemy Peredachi
  Informatsii}, vol.~28, no.~3, pp. 3--26, 1992.

\bibitem{Dai95}
J.~G. Dai, ``On positive harris recurrence of multiclass queueing networks: a
  unified approach via fluid limit models,'' \emph{The Annals of Applied
  Probability}, pp. 49--77, 1995.

\bibitem{DaiMeyn}
J.~G. Dai and S.~P. Meyn, ``Stability and convergence of moments for multiclass
  queueing networks via fluid limit models,'' \emph{IEEE Transactions on
  Automatic Control}, vol.~40, no.~11, pp. 1889--1904, 1995.

\bibitem{MeynBook}
S.~Meyn, \emph{Control techniques for complex networks}.\hskip 1em plus 0.5em
  minus 0.4em\relax Cambridge University Press, 2008.

\bibitem{Williams2}
R.~J. Williams, ``Diffusion approximations for open multiclass queueing
  networks: sufficient conditions involving state space collapse,''
  \emph{Queueing systems}, vol.~30, no. 1-2, pp. 27--88, 1998.

\bibitem{BillBook}
P.~Billingsley, \emph{Convergence of probability measures}.\hskip 1em plus
  0.5em minus 0.4em\relax John Wiley \& Sons, 2013.

\bibitem{HarrisonBook}
J.~Harrison, ``Brownian motion and stochastic flow systems,'' 1985.

\bibitem{StolyarSSC}
A.~L. Stolyar \emph{et~al.}, ``Maxweight scheduling in a generalized switch:
  State space collapse and workload minimization in heavy traffic,'' \emph{The
  Annals of Applied Probability}, vol.~14, no.~1, pp. 1--53, 2004.

\bibitem{BramsonSSC}
M.~Bramson, ``State space collapse with application to heavy traffic limits for
  multiclass queueing networks,'' \emph{Queueing Systems}, vol.~30, no. 1-2,
  pp. 89--140, 1998.

\bibitem{GamarZeevi}
D.~Gamarnik, A.~Zeevi \emph{et~al.}, ``Validity of heavy traffic steady-state
  approximations in generalized jackson networks,'' \emph{The Annals of Applied
  Probability}, vol.~16, no.~1, pp. 56--90, 2006.

\bibitem{Budhiraja}
A.~Budhiraja and C.~Lee, ``Stationary distribution convergence for generalized
  jackson networks in heavy traffic,'' \emph{Mathematics of Operations
  Research}, vol.~34, no.~1, pp. 45--56, 2009.

\bibitem{PaperICC}
A.~K. KS and V.~Sharma, ``Distributed control and quality-of-service in
  multihop wireless networks,'' in \emph{2018 IEEE International Conference on
  Communications (ICC)}.\hskip 1em plus 0.5em minus 0.4em\relax IEEE, 2018, pp.
  1--7.

\bibitem{BassBook}
R.~F. Bass, \emph{Stochastic processes}.\hskip 1em plus 0.5em minus 0.4em\relax
  Cambridge University Press, 2011, vol.~33.

\bibitem{GutBook}
A.~Gut, \emph{Stopped random walks}.\hskip 1em plus 0.5em minus 0.4em\relax
  Springer, 2009.

\end{thebibliography}
\bibliographystyle{IEEEtran}
\appendix
\subsection{Properties of $\hat v$}
	Here, we will complete the proof of Theorem \ref{CompleteTheorem} by showing that, along the subsequence $\mathcal{N}_2$, we have a limit $\hat v$ of $\hat v^n$, which has the properties:
	\begin{enumerate}
		\item $\hat v(t)$ is continuous.
		\item $\hat v(t)$ is finite for $t\in[0,\infty)$
		\item $\hat v(0)=0$
		\item 	If $\hat{w}(t)>0$, then $t$ is not a point of increase of $\hat{v}$.
	\end{enumerate}
To prove these properties, we need to study a set of fluid sample paths.\\
\subsection{Rescaled Fluid Paths}
To study diffusion properties on an interval $[t_n,t_n+\delta]$ for $\delta>0$, we look at fluid paths on the time $[nt_n,nt_n+n\delta]$. We consider the following family of fluid paths, started at a time $T$ apart from each other. For a time evolving process $f(t)$, define the operator $\Theta(\tau)$ as the shift, corresponding to the process started at time $\tau$. 

Consider the fluid scaled process  $z^n$. Consider a shifted form of these processes,
\begin{align}
\tilde z^{m,l}=\Theta(mt_m+Tl)z^m,
\end{align}
where $\Theta(x)f$ denotes the function $f$ started at $x$. Define the family of processes,
\begin{align}
\mathcal{Z}=\{\tilde z^{m,l(m)},m\in\mathcal{N}_3\},
\end{align}
where the index set $\mathcal{N}_3$ has the property that as $m\to\infty$ along $\mathcal{N}_3$, $t_m\to t$. Using these fluid paths we can obtain properties of the diffusion scaled process, since an interval of time $[mt,mt+m\delta]$ on the diffusion scale corresponds to a time $[t,t+\delta]$ on the diffusion scale.

If $t_m\to t$, and $l(m)\in[0,2\delta m/T-1]$, a time $s\in[0,T]$ for the path $\tilde z(m,l(m))$, for $m$ large enough, corresponds to a time,
\begin{align}
s^{'}=t_m+l(m)T/m+s/m \in [t-3\delta,t+3\delta]^+.
\end{align}
We have the following results regarding the behaviour of the fluid sample paths, from \cite{StolyarSSC}. The first is presented without proof.
\begin{lem}\label{LemmaOrigi}
	Consider the family $\mathcal Z$ with an associated sequence $t_m$, constants $T$ and $\delta$, both positive. Assume that $|\tilde q^{m,l(m)}|\in[c_1,c_2]$, with $0\leq c_1\leq c_2<\infty$, and $l(m)\in[0,2\delta m/T-1]\cap \mathbb{Z}$. Then, there is a subsequence $m_k$ along which, $\tilde z^{m,l(m)}$ converges to a fluid limit $z$, u.o.c, with $|q(0)|\in[c_1,c_2]$. 
\end{lem}
For a fluid limit $q(t)$, define the Lyapunov function,
\begin{align}
\mathcal L_1(q(t))=-\int_t^{\infty}\exp(t-\tau)\sum_{i,f} \alpha(q^f(\tau))q_{i}^f(\tau)\dot{q}_i^f(\tau) d\tau. \label{LyapunovFuncOne}
\end{align}
It can be shown \cite{PaperICC} that this function is non negative, finite and its time derivative  is negative. If, along $q(t)$, if $\lim_{t\to\infty}\mathcal L_1(q(t))=0$, define $\mathcal L_3=\mathcal L_1$. Else, if $\lim_{t\to\infty}\mathcal L_1(q(t))=\mathcal L_*>0$, define $\mathcal L_3(q(t))=\frac{\mathcal L_1(q(t))}{\mathcal L_*}-1$. Clearly, $\mathcal L_3(q(t))$ decreases to zero along any fluid path. Then, we have the following result.
\begin{lem}\label{RestFluidThm}
	Under our scheduling policy, assume that there is a subsequence such that, along this, $\hat v^n\to\hat v$. Suppose further that along this subsequence, we have
	\begin{align}
	s_m &\to s\geq 0, \hat{w}^m(s_m)\to K>0, \\
	 \limsup_{m\to\infty} &|\hat{q}^m(s_m)|<K_1K,
	\end{align}
	for some fixed $K_1>1$. Let $\delta>0$ be chosen such that,
	\begin{align}
	\epsilon=O_{\hat u}([s-3\delta,s+3\delta]^+)<0.5K,
	\end{align}
	where $O_{\hat u}[a,b]=\sup_{x,y\in[a,b]}|u(x)-u(y)|$. 	Let $K_2=\beta^2K_1K+2\epsilon$. Then, for any $\epsilon_2>0$ sufficiently small, there exists a time $T$ such that, for $m$ sufficiently large, we have,
	\begin{align}
	K-2\epsilon &<\tilde w^{m,0}(u)<K_2,\text{ for }\ u\in[0,T],\\
	(K-2\epsilon)/{\beta} &<|\tilde q^{m,0}(u)|<2\beta K_2 .
	\end{align}
	For $l\in [1,2\delta rT^{-1}-1]\cap\mathbb{Z}$, we have,
	\begin{align}
	\mathcal L_3(\tilde q^{m.l}(0)) &<2\epsilon_2,\label{OneKey}\\
	\mathcal L_3(\tilde q^{m.l}(T)) &<2\epsilon_2, \label{TwoKey}\\
	\mathcal L_3(\tilde q^{m.l}(u)) &<3\epsilon_2,\text{ for }\ u\in[0,T],\label{ThreeKey}\\
	\tilde v^{m,l}(u) &=\tilde v^{m,l}(u)-\tilde v^{m,l}(0)=0,\text{ for }\ u\in[0,T],\label{FourKey}\\
	K-2\epsilon &<\tilde w^{m,l}(u)<K_2,\text{ for }\ u\in[0,T],\label{FiveKey}\\
	(K-2\epsilon)/{\beta} &<|\tilde q^{m,l}(u)|<2\beta K_2 ,\label{SixKey}
	\end{align}
\end{lem}
The proof of this Lemma is an adaptation of the proof of Lemma 7 \cite{StolyarSSC} to our case. We present the main arguments below.
\begin{proof}[Proof of Lemma \ref{LemmaOrigi}]
	Observe that, since $\mathcal L_3$ is decreasing to zero, there exists a time $T$, such that,
	\begin{align}
	\mathcal L_3(t)\leq \epsilon_2,\ \forall t\geq T.
	\end{align}
	Consider the case $l=0$.
	First, observe that, for $m$ large enough,
	\begin{align}
	\limsup_{m\to\infty}\sup_{u\in[0,T]}|\tilde q^{m,0}(u)|<\beta\limsup_{m\to\infty}|\tilde q^{m,0}(u)|
	\end{align}
	This is true because, if it were not, using Lemma \ref{LemmaOrigi},  we could have a sequence of $\tilde z^{m,0}$ which converge to a fluid limit $z$ with  $|q(u)|\geq\beta|q(0)|$ for some $u$. However, this is not possible since,
	\begin{align}
	\sup_{t\geq 0}|q(t)|<\beta|q(0)|.
	\end{align}
	Alongwith our assumptions on $m$, this implies that,
	\begin{align}
	\limsup_{m\to\infty}\sup_{u\in[0,T]}|\tilde q^{m,0}(u)| &<\beta\limsup_{m\to\infty}|\tilde q^{m,0}(u)|<\beta K_1K, \\
	\limsup_{m\to\infty}\sup_{u\in[0,T]}\tilde w^{m,0}(u) &<\beta^2 K_1K.
	\end{align}
	Using the non decreasing property of $w$, we can show,
	\begin{align}
	\liminf_{m\to\infty}\inf_{u\in[0,T]}\tilde w^{m,0}(u)\geq K.
	\end{align}
	Choosing $T$ large enough, we can have,
	\begin{align}
	\mathcal L_3(\tilde q^{m,0}(T))<2\epsilon_2.
	\end{align}
	Since $\tilde q^{m,0}(T)=\tilde q^{m,1}(0)$, it also follows that,
	\begin{align}
	\mathcal L_3(\tilde q^{m,1}(0))<2\epsilon_2.
	\end{align}
	Now, consider the following properties, for $l\in[1,2\delta m/T-1]$.
	\begin{align}
	\mathcal L_3(\tilde q^{m.l}(0)) &<2\epsilon_2,\label{OneKey1}\\
	\mathcal L_3(\tilde q^{m.l}(T)) &<2\epsilon_2, \label{TwoKey1}\\
	\mathcal L_3(\tilde q^{m.l}(u)) &<3\epsilon_2,\text{ for }\ u\in[0,T],\label{ThreeKey1}\\
	\tilde v^{m,l}(u) &=\tilde v^{m,l}(u)-\tilde v^{m,l}(0)=0,\text{ for }\ u\in[0,T],\label{FourKey1}\\
	K-2\epsilon &<\tilde w^{m,l}(u)<K_2,\text{ for }\ u\in[0,T],\label{FiveKey1}\\
	(K-2\epsilon)/{\beta} &<|\tilde q^{m,l}(u)|<2\beta K_2 .\label{SixKey1}
	\end{align}
	We will show these hold, by induction on $l$. Asssume the properties hold for all $l<l_1$, but at least one of the abover properties is violated for $l=l_1$. Since the properties hold up to $l=l_1-1$, we have that,
	\begin{align}
	\mathcal L_3(\tilde q^{m,l_1}(0))=\mathcal L_3(\tilde q^{m,l_1-1}(T))<2\epsilon_2.
	\end{align}
	Since $w$ is non decreasing, we have,
	\begin{align}
	\tilde w^{m,l_1}(0)>K-2\epsilon.
	\end{align}
	From the relation between $|q|$ and $w$ it follows that,
	\begin{align}
	|\tilde q^{m,l_1}(0)|\in\left[\frac{K-2\epsilon}{\beta},2\beta K_1\right].
	\end{align}
	Thus, for a choice of $T$ appropriately large, we will have,
	\begin{align}
	\mathcal L_3(\tilde q^{m.l_1}(0)) &<2\epsilon_2,\\
	\mathcal L_3(\tilde q^{m.l_1}(T)) &<2\epsilon_2, \\
	\mathcal L_3(\tilde q^{m.l_1}(u)) &<3\epsilon_2,\text{ for }\ u\in[0,T]. \label{smallLyapEq}
	\end{align}
	To show the non-increasing property of $\tilde v$ as in (\ref{FourKey1}), observe that the queue length and workload are strictly positive as shown above. Since we had,
	\begin{align}
	v^{m,l}(t)=x^{m,l}(t)-\langle\psi,d^{m,l}(t)-r^{m,l}(t)\rangle,
	\end{align}
	and since our optimization is such that we choose the allocation vector ${\mu}^*$ such that,
	\begin{align}
	\mu^* &=\arg_{\mu}\max \sum_{i,j,f}\alpha(q_i^f)q_{ij}^f \mu_{ij}^f,\\
	&= \arg_{\mu}\max \sum_{i,j,f}\alpha(q_i^f)(q_i^f-q_j^f)\mu_{ij}^f.
	\end{align}
	The second equation holds because the allocation vector $\dot s_{ij}^f(t)$ is zero when $q_i^f-q_j^f\leq 0$. This optimization may be rewritten as a function of new variables $\tilde{\mu}$, where $\tilde{\mu}_i^f=\sum_j\mu_{ij}^f-\sum_k\mu_{ki}^f$. We have the optimal $\tilde{\mu}^*$ given by, 
	\begin{align}
	\tilde{\mu}^*= \arg_{\tilde{\mu}}\max \sum_{i,j,f}\alpha(q_i^f)q_i^f\tilde{\mu}_{i}^f.
	\end{align}
	Since (\ref{smallLyapEq}) holds, it will be that (choosing $\epsilon_2$ small enough), this is exactly the result of the optimization,
	\begin{align}
	\tilde{\mu}^*= \arg_{\tilde{\mu}}\max \sum_{i,j,f}\psi_i^f\tilde{\mu}_{i}^f,
	\end{align}
	since the function $\mathcal L_3$ indicates how close we are to the collapse vector $\psi$. From the definition of $X$, it follows that the scaled $\tilde x$ attains the value given above, and hence $\tilde v$ does not increase in the interval.
	
	Since $\tilde v$ remains at zero, we can see that any increase in $\tilde w$ is an increase in $\tilde u$, and hence,
	\begin{align}
	\tilde w^{m,l_1}(u)=\tilde w^{m,0}(T)+\tilde u^m(t_m+l_1T/m+u/m)-\tilde u^m(t_m+T/m).
	\end{align}
	Since the oscillation of $\hat u$ is bounded and since $\tilde u^m\to\hat u$, the bounds (\ref{FiveKey1}) and (\ref{SixKey1}) also follow for $l_1$. Hence, we have inductively shown that the properties (\ref{OneKey})-(\ref{SixKey}) hold.
\end{proof}
With the above result, we also obtain the properties of $\hat v$.
\subsection{Proof of the properties of $\hat v$}
	The proof of this result follows as an application of  Lemma \ref{RestFluidThm}, as in the proof of Theorem 1 in \cite{StolyarSSC}.
\end{document}